%% file: main_arxiv2.tex
\documentclass{article}
  \usepackage{authblk}
  \usepackage{fullpage}
 \usepackage[ruled,vlined]{algorithm2e}
  \usepackage{amsmath,amsfonts}
  \usepackage{amsthm}
  \usepackage{authblk}
  \usepackage[utf8]{inputenc}
  \usepackage{comment}
  \usepackage{forest}
  \usepackage{xspace}
  \usepackage{enumerate}
  \usepackage{todonotes}
  \usepackage[hidelinks]{hyperref}
  \usepackage[capitalise]{cleveref}
  \usepackage{graphics,adjustbox}
  \usepackage{tabularx}
  \usepackage{amssymb}
  \usepackage{lscape}
  \usepackage{color}
  \usepackage{tikz}
\usetikzlibrary{decorations.pathmorphing,decorations.markings,trees, arrows, shapes}
\usepackage{enumitem}

\newtheorem{theorem}{Theorem}
\newtheorem{lemma}[theorem]{Lemma}

\newtheorem{fact}[theorem]{Fact}
\newtheorem{corollary}[theorem]{Corollary}
\newtheorem{observation}[theorem]{Observation}
\newtheorem{claim}[theorem]{Claim}

\theoremstyle{remark}
\newtheorem{remark}[theorem]{Remark}
\newtheorem{example}[theorem]{Example}

\newcommand{\dd}{\mathinner{.\,.}}
\newcommand{\Oh}{\mathcal{O}}

\newcommand{\Label}{\mathcal{L}}

\newcommand{\lcp}{\textsf{lcp}}

\newcommand{\lce}{\textsf{lce}}

\newcommand{\MS}{\textit{MS}}
\newcommand{\MP}{\operatorname{A}}
\newcommand{\countop}{\mathit{Count}}
\newcommand{\R}{\mathbf{R}}

\newcommand{\val}{\mathit{val}}
\newcommand{\idx}{\mathit{idx}}

\newlength{\unit}
\newlength{\nodemargin}
\newlength{\nodesep}
\newlength{\level}

 \newcommand{\defproblem}[3]{
  \vspace{2mm}
\noindent\fbox{
  \begin{minipage}{0.96\textwidth}
  #1\\
  {\bf{Input:}} #2  \\
  {\bf{Output:}} #3
  \end{minipage}
  }
  \vspace{2mm}
}

\begin{document}
\title{Efficient Computation of Sequence Mappability}
\author[1]{Panagiotis Charalampopoulos}
\author[2]{Costas S.\ Iliopoulos}
\author[3]{Tomasz Kociumaka}
\author[4,5]{Solon P.\ Pissis}
\author[6]{Jakub Radoszewski\thanks{Supported by the ``Algorithms for text processing with errors and uncertainties'' project carried out within the HOMING programme of the Foundation for Polish Science co-financed by the European Union under the European Regional Development Fund.}}
\author[6]{Juliusz Straszyński\protect\footnotemark[1]}

\affil[1]{\normalsize
Efi Arazi School of Computer Science, The Interdisciplinary Center Herzliya, Herzliya, Israel\\
    \texttt{panagiotis.charalampopoulos@post.idc.ac.il} 
}
\affil[2]{\normalsize
    Department of Informatics, King's College London, London, UK\\
    \texttt{c.iliopoulos@kcl.ac.uk} 
}
\affil[3]{\normalsize
      University of California, Berkeley, USA\\
      \texttt{kociumaka@berkeley.edu} 
}
\affil[4]{\normalsize
      CWI, Amsterdam, The Netherlands\\
      \texttt{solon.pissis@cwi.nl} 
}
\affil[5]{\normalsize
      Vrije Universiteit, Amsterdam, The Netherlands
}
\affil[6]{\normalsize
      Institute of Informatics, University of Warsaw, Warsaw, Poland\\
      \texttt{[jrad,jks]@mimuw.edu.pl} 
}

\date{\vspace{-1cm}}

\maketitle 
\begin{abstract}
Sequence mappability is an important task in genome resequencing.
In the $(k,m)$-mappability problem, for a given sequence $T$ of length $n$, the goal is to compute a table whose $i$th entry is the number of indices $j \ne i$ such that the length-$m$ substrings of $T$ starting at positions $i$ and $j$ have at most $k$ mismatches.
Previous works on this problem focused on heuristics computing a rough approximation of the result or on the case of $k=1$. 
We present several efficient algorithms for the general case of the problem. Our main result is an algorithm that,
for $k=\Oh(1)$, works in $\Oh(n)$ space and, with high probability, in $\Oh(n \cdot \min\{m^k,\log^k n\})$ time.
Our algorithm requires a careful adaptation of the $k$-errata trees of Cole et al.~[STOC 2004] to avoid multiple counting of pairs of substrings. 
Our technique can also be applied to solve the all-pairs Hamming distance problem introduced by Crochemore et al.~[WABI 2017].
We further develop $\Oh(n^2)$-time algorithms to compute \emph{all} $(k,m)$-mappability tables for a fixed $m$ and all $k\in \{0,\ldots,m\}$ or a fixed $k$ and all $m\in\{k,\ldots,n\}$. 
Finally, we show that, for $k,m = \Theta(\log n)$, the $(k,m)$-mappability problem cannot be solved in strongly subquadratic time unless the Strong Exponential Time~Hypothesis~fails.

This is an improved and extended version of a paper that was presented at SPIRE 2018.
\end{abstract}

\section{Introduction}
\paragraph{The $k$-mappability problem.}

Analyzing data derived from massively parallel sequencing experiments often
depends on the process of genome assembly via resequencing; namely, assembly with the help of a reference sequence. In this process, a large number of reads (or short sequences) derived from a DNA donor during these experiments must be mapped back to a reference sequence, comprising a few gigabases, to establish the section of the genome from which each read has been derived. An extensive number of short-read alignment techniques and tools have been introduced to address this challenge emphasizing on different aspects of the process~\cite{Fonseca}. 
 
In turn, the process of resequencing depends heavily on how mappable a genome is with respect to reads of some fixed length $m$. Thus, given a reference sequence, for every substring of length $m$ in the sequence, we want to count how many additional times this substring appears in the sequence when allowing for a small number $k$ of errors. This computational problem and a heuristic approach to approximate the solution were first proposed in~\cite{biopaper} (see also~\cite{ITAB2009}), where a great variance in genome mappability between species and gene classes was revealed.

More formally, for a string $T$, let $T_i^m$ denote the length-$m$ substring of $T$ that starts at position $i$. In the $(k,m)$-mappability problem, for a given string $T$ of length $n$, we are asked to compute a table $\MP^m_{\le k}$ whose $i$th entry $\MP^m_{\le k}[i]$ is the number of indices $j \ne i$ such that the substrings $T_i^m$ and $T_j^m$ are at Hamming distance at most $k$. In the previous study~\cite{biopaper}, the assumed values of parameters were $k \le 4$, $m \le 100$, and the alphabet of $T$ was $\{\mathtt{A},\mathtt{C},\mathtt{G},\mathtt{T}\}$.

\begin{example}\label{ex}
  Consider a string $T=\texttt{aababba}$ and $m=3$.
  The following table shows the $(k,m)$-mappability counts for $k=1$ and $k=2$.
  
  \begin{table}[htbp]
  \begin{center}
  \renewcommand{\arraystretch}{1.2}
    \begin{tabular}{rc|c|c|c|c|c}
      \textbf{position} & $i$ & 1 & 2 & 3 & 4 & 5 \\\hline
      \textbf{substring} & $T_i^3$ & $\texttt{aab}$ & $\texttt{aba}$ & $\texttt{bab}$ & $\texttt{abb}$ & $\texttt{bba}$ \\\hline
      \textbf{$(1,3)$-mappability}&$\MP_{\le 1}^3[i]$ & 2 & 2 & 1 & 2 & 1 \\\hline
      \textbf{$(2,3)$-mappability}&$\MP_{\le 2}^3[i]$ & 3 & 3 & 3 & 4 & 3 \\\hline
      \textbf{difference}& $\MP_{=2}^3[i]$ & 1 & 1 & 2 & 2 & 2
    \end{tabular}
  \end{center}
  \end{table}
  
For instance, consider the position 1. The $(1,3)$-mappability is 2 due to the occurrences of $\texttt{bab}$ and $\texttt{abb}$ at positions 3 and 4, respectively. The $(2,3)$-mappability is 3 since only the substring $\texttt{bba}$, occurring at position 5, has three mismatches with $\texttt{aab}$.
\end{example}

For convenience, our algorithms compute an array $\MP_{=k}^m$ whose $i$th entry $\MP_{=k}^m[i]$ is the number of positions $j\ne i$
such that substrings $T_i^m$ and $T_j^m$ are at Hamming distance \emph{exactly} $k$.
Note that $\MP_{\le k}^m[i]=\sum_{\kappa=0}^{k} \MP_{=\kappa}^m[i]$; see the ``difference'' row in the example above. Henceforth, we call this problem \emph{the $(k,m)$-mappability problem}.

\begin{table}[htbp]
    \centering
    \begin{tabular}{c|c}
        \textbf{Solution} & \textbf{Time complexity} \\\hline
        Manzini~\cite{DBLP:conf/spire/Manzini15} & $\Oh(mn \log n / \log\log n)$ \\
        Alzamel et al.~\cite{DBLP:conf/cocoa/AlzamelCIPRS17} & $\Oh(nm)$ \\
        Alzamel et al.~\cite{DBLP:conf/cocoa/AlzamelCIPRS17} & $\Oh(n \log n \log \log n)$ \\
        Alzamel et al.~\cite{DBLP:conf/cocoa/AlzamelCIPRS17} & $\Oh(n)$ on average for $m=\Omega (\log n)$ \\
        Amir et al.~\cite{barilan_kmap}, Hooshmand et al.~\cite{DBLP:conf/iccabs/HooshmandAGAT18} & $\Oh(n \log n)$ \\
        Amir et al.~\cite{barilan_kmap} & $\Oh(n)$ for $m=\Omega (\sqrt{n})$\\
    \end{tabular}
    \caption{Known algorithms for computing $(1,m)$-mappability for strings over constant-sized alphabets. All algorithms use $\Oh(n)$ space.}%
    \label{tab:1map}
\end{table}

Using the suffix array and the LCP table~\cite{DBLP:journals/siamcomp/ManberM93,DBLP:conf/cpm/KasaiLAAP01,DBLP:journals/jacm/KarkkainenSB06}, the $(0,m)$-mappability problem can be solved in $\Oh(n)$ time and space. Known solutions for computing $(1,m)$-mappability are shown in Table~\ref{tab:1map}; the $\Oh(nm)$-time and the $\Oh(n)$-average-time solutions of Alzamel et al.~\cite{DBLP:conf/cocoa/AlzamelCIPRS17} work also on strings over \emph{integer alphabets} $\{1,\dots,\sigma\}$ for $\sigma = n^{\Oh(1)}$. Moreover, the latter algorithm was shown to be generalizable to arbitrary $k$, requiring $\Oh(n)$ space and, on average, $\Oh(kn)$ time if $m=\Omega (k\log_\sigma n)$. 
In~\cite{DBLP:conf/sofsem/AlamroACIP18}, the authors introduced an efficient construction of a \emph{genome mappability array} $B_k$ in which $B_k[\mu]$ is the smallest length $m$ such that at least $\mu$ of the length-$m$ substrings of $T$ do not occur elsewhere in $T$ with at most $k$ mismatches.
This construction was further improved in~\cite{DBLP:conf/spire/AyadBCIP18}.

\paragraph{The all-pairs Hamming distance problem.}

The evolutionary relationships between different species or taxa are usually inferred
through phylogenetic analysis techniques. Some of these techniques rely on the
inference of phylogenetic trees. A first step of these techniques is to compute the distances between all pairs of sequences representing the set of species or taxa under study. This particular step, however, often dominates the running time of these methods. Depending on the application, the underlying model of evolution, and the optimality criterion, it may not be strictly necessary to be aware of the complete distance matrix (see~\cite{Francisco2009,crochemore_et_al:LIPIcs:2017:7652}, for instance). Thus, in this preprocessing step, we are only interested in pairs with distances not exceeding a given threshold. 

The computational problem can be formally defined as follows. Given a set $\R$ of $r$ length-$m$ strings and an integer $k\in\{0,\ldots, m\}$, return all pairs $(X_1,X_2)\in \R\times \R$, with $X_1\ne X_2$, such that $X_1$ and $X_2$ are at Hamming distance at most $k$. This problem has been studied in the average-case model and efficient linear-time algorithms are known under some constraints on the value of $k$ and some assumptions on the elements of $\R$~\cite{crochemore_et_al:LIPIcs:2017:7652,MAKINEN201917,Grabowski}. The indexing variant of the all-pairs Hamming distance problem has further applications in bioinformatics for querying typing databases~\cite{Carrico2018} and in information retrieval for searching similar documents in a collection~\cite{Gog:2016:FCH:2911451.2911523}.

Intuitively, the connection between the $(k,m)$-mappability problem and the all-pairs Hamming distance problem is as follows. By first concatenating the $r$ elements of $\R$ to construct a new string $T$ of length $n=rm$, solving the former considering only the $r$ substrings of $T$ starting at positions $i$, with $i \text{ mod } m=1$, and summing up the resulting values, we would obtain the total size of the output of the latter.

Henceforth we assume, as in the mappability problem, that we are to compute all pairs at Hamming distance \emph{exactly} $k$. In the end, we run the algorithm for all values of $k$ up to a given threshold of interest.

\paragraph{Our contributions.}
We present several algorithms for the general case of the $(k,m)$-mappability problem. 
More specifically, our contributions are as follows:
\begin{enumerate}
\item In Section~\ref{sec:hp}, we show a randomized Las-Vegas algorithm for the $(k,m)$-mappa\-bility problem that works in $\Oh(n\binom{\log n+k}{k}4^kk)$ time with high probability ($1-n^{-c}$ for an arbitrarily large constant parameter $c>0$) and $\Oh(n2^kk)$ space for a string over an ordered alphabet.\label{cont:alg}
It requires a careful adaptation of the technique of recursive heavy-path decompositions in a tree~\cite{DBLP:conf/stoc/ColeGL04}.
\item In Section~\ref{sec:2.5}, we show an algorithm to solve all-pairs Hamming distance problem in time $\Oh(rm+ r\binom{\log r+k}{k}4^kk \log r + \mathsf{output}\,2^kk \log r)$ and space $\Oh(rm+r2^kk \log r)$.\label{cont:allp}
\item In Section~\ref{sec:2}, we show an algorithm for the $(k,m)$-mappability problem that works in $\Oh(n k \cdot (m+1)^k)$ time and $\Oh(n)$ space for a string over an integer alphabet. Together with the first result, this yields an $\Oh(n \cdot \min\{m^k,\log^k n\})$-time and $\Oh(n)$-space algorithm for $k=\Oh(1)$.\label{cont:mk}
\item In Section~\ref{sec:3}, we show $\Oh(n^2)$-time algorithms to compute \emph{all} $(k,m)$-mappability tables for a fixed $m$ and all $k\in\{0,\ldots,m\}$, or for a fixed $k$ and all $m\in\{k,\ldots,n\}$.
\item Finally, in Section~\ref{sec:4}, we prove that the $(k,m)$-mappability problem for $k,m = \Theta(\log n)$ cannot be solved in strongly subquadratic time unless the Strong Exponential Time Hypothesis~\cite{DBLP:journals/jcss/ImpagliazzoPZ01,DBLP:journals/jcss/ImpagliazzoP01} fails.\label{cont:lb}
\end{enumerate}
In contributions~\ref{cont:alg} and~\ref{cont:lb}, we apply recent advances in the Longest Common Substring with $k$ Mismatches problem that were presented in~\cite{DBLP:conf/cpm/Charalampopoulos18} and~\cite{DBLP:journals/corr/abs-1712-08573}, respectively (see also~\cite{DBLP:journals/jcb/ThankachanAA16}). 
In particular, compared to~\cite{DBLP:conf/cpm/Charalampopoulos18}, our contribution~\ref{cont:alg} requires a careful counting of substring pairs to avoid multiple counting and a thorough analysis of the space usage. Technically this is the most involved contribution.

This work is an extended version of~\cite{DBLP:conf/spire/AlzamelCIKPRS18}. In comparison to the conference version, we improve the complexity of the main algorithm by a $\Theta(\log n)$-factor, remove the dependency 
on the alphabet size in contribution~\ref{cont:mk}, and apply our techniques to solve the all-pairs Hamming distance problem (contribution~\ref{cont:allp}).

\section{Preliminaries}

Let $T=T[1]T[2]\cdots T[n]$ be a \emph{string} of length $|T|=n$ over a finite ordered alphabet $\Sigma$ of size $|\Sigma|=\sigma$. For two positions $i$ and $j$ on $T$, the \emph{substring} (sometimes called \emph{factor}) of $T$ that starts at position $i$ and ends at position $j$ is $T[i]\cdots T[j]$ (it is of length $0$ if $j<i$). A \emph{prefix} of $T$ is a substring that starts at position 1 and a \emph{suffix} of $T$ is a substring that ends at position $n$. We denote the suffix that starts at position $i$ by $T_i$ and its prefix of length $m$ by $T_i^m$.

The \emph{Hamming distance} between two strings $S$ and $T$ of the same length $|S| = |T|$ is defined as $d_H(S, T) = |\{i\in \{1, 2,\ldots, |S|\} : S[i] \neq T[i]\}|$. If $|S| \neq |T|$, we set $d_H(S, T)=\infty$. 
  
By $\lcp(S,T)$ we denote the length of the longest common prefix of $S$ and $T$.
For a fixed string $T$, we also set $\lce(r, s)=\lcp(T_{r},T_{s})$.
By $\lce_k(r,s)$ we denote the length of the longest common prefix of $T_r$ and $T_s$ when up to $k$ mismatches are allowed, that is, the maximum $\ell$ such that $d_H(T_r^\ell,T_s^\ell) \le k$.

\paragraph{Compact trie.} A \emph{trie} of a collection of strings $C$ is a labeled tree that contains a node for every distinct prefix of a string in $C$; the root node is $\varepsilon$; the set of \emph{terminal} nodes is $C$; and edges are of the form $u\stackrel{c}{\to}uc$, where $u$ and $uc$ are nodes and $c\in\Sigma$. A compact trie $\mathbf{T}$ of a collection of strings $C$ is obtained from the trie of $C$ by dissolving all non-branching nodes, excluding the root and the terminals. The nodes of the trie which become nodes of $\mathbf{T}$ are called \emph{explicit} nodes, while the other nodes are called \emph{implicit}. Each edge of $\mathbf{T}$ can be viewed as an upward maximal path of implicit nodes starting with an explicit node. The string label of an edge is a substring of one of the strings in $C$; the label of an edge is the first letter of the edge's string label. Each node of the trie can be represented in $\mathbf{T}$ by the edge it belongs to and an index within the corresponding path. We let $\Label(v)$ denote the \emph{path-label} of a node $v$, i.e., the concatenation of the string labels of the edges along the path from the root to $v$. Additionally, $\mathbf{D}(v)= |\Label(v)|$ is the \emph{string-depth} of node~$v$.

\paragraph{Suffix tree.} The suffix tree of a string $T$ is the compact trie representing all suffixes of~$T$. 
%
%
The suffix tree of a string $T$ of length $n$ over an integer alphabet 
can be constructed in $\Oh(n)$ time~\cite{DBLP:conf/focs/Farach97}
and, after an $\Oh(n)$-time preprocessing~\cite{DBLP:conf/latin/BenderF02},
it can be used to answer $\lce(r,s)$ queries in $\Oh(1)$ time.

\paragraph{Hashing.} We use perfect hashing to implement dynamic dictionaries supporting insertions and deletions
of entries (key-value pairs), as well as to retrieve an arbitrary entry with a given key.
Technically, we maintain a single global dictionary (which may simulate multiple local dictionaries) implemented using~\cite[Theorem 1.1]{DBLP:conf/icalp/DietzfelbingerH90}.
In the preprocessing, we insert $n$ dummy entries; this incurs extra $\Oh(n)$ terms in the time and space complexities,
but also guarantees that the running time of every operation is $\Oh(1)$ with probability at least $1-n^{-c}$,
where $c > 0$ is a constant specified at initialization time.
As long as the total number of dictionary operations is polynomial in $n$, we derive 
Las-Vegas algorithms whose running times bounds hold with high probability (rather than just in expectation).
Whenever the time complexity of any algorithm in this work is superpolynomial in $n$ (which may happen for large values of $k$), we resort to naive polynomial-time solutions.

When using strings as dictionary keys,
we rely on Karp--Rabin fingerprints (polynomial hashing)~\cite{DBLP:journals/ibmrd/KarpR87} with collision probability bounded by $n^{-C}$ for strings of length at most $n$ (and a sufficiently large constant $C$).
In order to obtain Las-Vegas algorithms, we provide mechanisms for detecting collisions and resort to naive polynomial-time solutions upon detecting any.

\section{Computing Mappability in $\Oh(n \log ^k n)$ Time and $\Oh(n)$ Space}\label{sec:hp}

Our algorithm operates on so-called \emph{modified strings}. A modified string $\alpha$ is a string $U$ with a set of modifications $M$. Each element of the set $M$ is a pair of the form $(i,c)$ which denotes a substitution ``$U[i]:=c$''. We assume that no two pairs in $M$ share the same index $i$. By $\val(\alpha)$ we denote the string $U$ after all the substitutions and by $M(\alpha)$ we denote the set $M$.
The sets $M(\alpha)$ for modified strings are implemented as (functional) lists.
Whenever a modified string $\beta$ is obtained by introducing an extra modification to a modified string $\alpha$, the head of $M(\beta)$ represents the new modification whereas the tail points to $M(\alpha)$.
We always introduce modifications in the left-to-right order so that the lists $M(\alpha)$ are sorted according to the decreasing order of indices $i$.

The algorithm processes \emph{modified substrings} of $T$ that are modified strings originating from the substrings~$T_i^m$. For a modified substring $\alpha$ originating from $T_i^m$, we denote $\idx(\alpha)=i$.

\paragraph{Overview of the algorithm.}
Intuitively, the algorithm proceeds by efficiently simulating transformations
of a compact trie of modified substrings, initially containing all substrings $T_i^m$. 
The elementary transformations are guided by the \emph{smaller-to-larger} principle,
and each of them consists in copying one subtree unto its sibling, with an appropriate modification introduced to each copied substring in order to match the label of the edge leading to the sibling.
This process effectively results in registering one mismatch for a
large batch of substrings at once, and therefore lays a foundation to solve the main
problem in the aforementioned time.

More precisely, the algorithm navigates a compact trie of modified substrings.\footnote{The true course of the algorithm will not actually perform much of its
operations on a compact trie, but the intuition is best conveyed by
visualizing them this way.
}
The trie is constructed top-down recursively, and the final set of modified substrings that are present in the trie is known only when all the leaves of the trie have been reached.

In a recursive step, a node $v$ of the trie stores a set of modified substrings $\MS(v)$. Initially, the root $r$ stores all substrings $T_i^m$ in its set $\MS(r)$. The path-label $\Label(v)$ is the longest common prefix of (the values of) all the modified substrings in $\MS(v)$ and the string-depth $\mathbf{D}(v)$ is the length of this prefix. None of the strings in $\MS(v)$ contains a modification at a position greater than $\mathbf{D}(v)$. The children of $v$ are determined by subsets of $\MS(v)$ that correspond to different letters at position $\mathbf{D}(v)+1$. Furthermore, additional modified substrings with modifications at position $\mathbf{D}(v)+1$ are created and inserted into the children's $\MS$-sets. This corresponds to the intuition of copying subtrees unto their siblings.

The goal is to distribute the modified substrings into leaves and, by processing each leaf independently,
register exactly once every pair of substrings $(T_i^m, T_j^m)$ differing on exactly $k$ positions.

Now, we will describe the recursive routine for visiting a node.

\paragraph{Processing an internal node.}
Assume that our node $v$ has children $u_1,\ldots,u_a$. First, we distinguish a child of $v$ with maximum-size set $\MS$; let it be $u_1$. We will refer to this child as \emph{heavy} and to every other as \emph{light}. We will recursively branch into each child to take care of all pairs of modified substrings contained in any single subtree.

For this, we create an extra child $u_{a+1}$ so that $\MS(u_{a+1})$ contains all modified substrings from $\MS(u_2) \cup \dots \cup \MS(u_a)$ with the letters at position $\mathbf{D}(v)+1$
replaced by a common wildcard character~\$. By processing the subtree of $u_{a+1}$, we will consider pairs of modified substrings that originate from different light children.

Additionally, we insert all modified substrings from $\MS(u_2) \cup \dots \cup \MS(u_a)$ into $\MS(u_1)$, substituting the letter at position $\mathbf{D}(v)+1$ with the common letter at this position of modified substrings in $\MS(u_1)$. This transformation will take care of pairs between the heavy child and the light ones. 

Finally, the algorithm branches into the subtrees of $u_1,\ldots,u_{a+1}$. A pseudocode of this process is presented as Algorithm~\ref{algo}. Note that in the special case of a binary alphabet the child $u_{a+1}$ need not be created.
Moreover, since modified substrings with more than $k$ substitutions are irrelevant for our algorithm,
we refrain from creating them in the interest of time and space complexity.

\newcommand{\stmt}[1]{$ #1 $ \\}
\newcommand{\assign}[2]{\stmt{#1 \gets #2}}
\newcommand{\equal}[2]{$ #1 = #2 $}
\newcommand{\myforeach}[2]{\ForEach{$ #1 \in #2 $}}

\SetKwFunction{FR}{processNode}
\SetKwProg{Fn}{Procedure}{}{end}
\SetKwData{Aligned}{depth}
\SetKwData{FRP}{processNode}
\SetKwData{Tree}{$v$}
\SetKwData{Light}{light}
\SetKwData{SonsList}{children}
\SetKwData{SubstringLength}{m}
\SetKwData{HeavySon}{heavyChild}
\SetKwData{HeavyLetter} {heavyLetter}
\SetKwData{LightSon}{lightChild}
\SetKwData{MergedTree}{wildcardTree}
\SetKwData{Hash} {\$}
\SetKwData{ModifiedString} {$\alpha$}
\SetKwData{LightString} {$\alpha'$}
\SetKwData{LightStringTwo} {$\alpha''$}
\SetKwData{Index} {index}
\SetKwData{LightSons} {lightChildren}
\SetKwData{Son} {child}

\SetKwFunction{Begin}{least}
\SetKwFunction{Lcp}{lcp}
\SetKwFunction{ProcessLeaf}{processLeaf}
\SetKwFunction{SplitBy} {splitByLetter}
\SetKwFunction{PollHeaviest} {findHeaviest}
\SetKwFunction{Insert} {insert}

\SetKwInput{TreeInput}{\Tree}
\SetKwInput{InsertInput}{\Insert{\Tree, \ModifiedString}}
\SetKwInput{AlignedInput}{\Aligned}
\SetKwInput{LcpInput}{\Lcp{\Tree}}
\SetKwInput{SplitInput}{\SplitBy{\Tree, \Index}}

\begin{algorithm}[ht]
\caption{A recursive procedure of processing a trie node}\label{algo}

\Fn{\FR{\Tree}}{
\LcpInput{computes the longest common prefix of all the strings in $\MS(\Tree)$}
\InsertInput{inserts \ModifiedString into $\MS(\Tree)$}
\SplitInput{splits $\MS(\Tree)$ into groups having the same \Index-th letter,
returning a list of sets of modified substrings}
\vspace{0.3cm}

  \assign{\Aligned} {\Lcp(\Tree)}
  \If{\equal{\Aligned} {\SubstringLength}} {
    \stmt{\ProcessLeaf{\Tree}}
    \Return\\
  }
  
  \vspace{0.1cm}
  \assign{\SonsList} {\SplitBy{\Tree, {$\Aligned+1$}}}
  \assign{\HeavySon} {\PollHeaviest{\SonsList}}
  \assign{\HeavyLetter} {\val(\alpha)[\Aligned$+1$]\text{ for some }\alpha \in \HeavySon}

  \assign{\MergedTree} {\emptyset}

  \myforeach{\LightSon}{\SonsList \setminus \{\HeavySon\}}{
    \myforeach{\ModifiedString}{\LightSon}{
      \If{$|M(\ModifiedString)| < k$}{
        \assign{\LightString}{\ModifiedString}
        \assign{\LightString[\Aligned$+1$]}{\Hash}
        \stmt{\Insert{\MergedTree, \LightString}}
        \assign{\LightStringTwo}{\ModifiedString}
         \assign{\LightStringTwo[\Aligned$+1$]}{\HeavyLetter}
         \stmt{\Insert{\HeavySon, \LightStringTwo}}
      }
    }
  }

  \myforeach{\Son}{\SonsList\cup\{\MergedTree\}}{
    \stmt{\FR{\Son}}
  }
}
\end{algorithm}

\paragraph{Processing a leaf.}
Each modified substring $\alpha$ stores its index of origin $\idx(\alpha)$ and the set of modifications $M(\alpha)$. As we have seen, the substitutions introduced in the recursion are of two types: of wildcard origin and of heavy origin. For a modified substring $\alpha$, we introduce a partition $M(\alpha)=W(\alpha)\cup H(\alpha)$ into modifications of these kinds. 
For every leaf $v$, the modified substrings $\alpha\in \MS(v)$ share the same value $\val(\alpha)$, and hence $W(\alpha)$ is also the same. Finally, by $W^{-1}(\alpha)$ we denote the set $\{(j,T_{\idx(\alpha)}^m[j]) :  (j,\$)\in W(\alpha)\}$. 
We call modified substrings $\alpha,\beta\in \MS(v)$ \emph{compatible} if they satisfy the following condition:
\begin{equation}\label{eq:cond}
H(\alpha) \cap H(\beta) = \emptyset,\;\;\; W^{-1}(\alpha) \cap W^{-1}(\beta) = \emptyset, \;\;\; |H(\alpha)|+|H(\beta)|+|W(\alpha)|=k.
\end{equation}
Intuitively, $\alpha$ and $\beta$ are compatible only if the positions of modifications in $M(\alpha)\cup M(\beta)$ do not contain any position $j$ such that $T^m_{\idx(\alpha)}[j]=T^m_{\idx(\beta)}[j]$.
As proved in Lemma~\ref{lem:corr} below, for every $\alpha \in \MS(v)$, we should increment $A_{=k}^m[\idx(\alpha)]$
for each compatible $\beta\in \MS(v)$.
We next show how to efficiently count these  modified substrings using the inclusion-exclusion principle and several precomputed values, as we cannot afford to count them naively.

For convenience, let $R(\alpha)$ denote the union of disjoint sets $H(\alpha)$ and $W^{-1}(\alpha)$. For a leaf $v$, let $\countop(s,B)$ denote the number of modified substrings $\beta\in\MS(v)$ such that $|H(\beta)|=s$ and $B \subseteq R(\beta)$. All the non-zero values are stored in a hash table. They can be generated by iterating through all the subsets of $R(\beta)$ for all modified substrings $\beta\in \MS(v)$; this costs $\Oh(2^kk|\MS(v)|)$ time and space. Finally, the result for a modified substring $\alpha$ can be computed using the following direct consequence of the inclusion-exclusion principle.

\begin{lemma}\label{lem:inc_exc}
The number of modified substrings $\beta\in\MS(v)$ that are compatible with a modified substring $\alpha\in \MS(v)$ is $\sum_{B \subseteq R(\alpha)} (-1)^{|B|} \countop(k-|M(\alpha)|,B)$.
\end{lemma}
\begin{proof}
First, let $h=k-|M(\alpha)|$. We want to count the modified substrings $\beta \in \MS(v)$ that satisfy $|H(\beta)|=h$ and $R(\alpha)\cap R(\beta)=\emptyset$.
For $(i,x) \in R(\alpha)$, let $A_{(i,x)}=\{\beta \in \MS(v): |H(\beta)|=h \text{ and } (i,x) \in R(\beta)\}$.
Then, we want to compute $\countop(h,\emptyset)-|\bigcup_{(i,x) \in R(\alpha)} A_{(i,x)}|$.
By the inclusion-exclusion principle we have 
\[\left|\bigcup_{(i,x) \in R(\alpha)} A_{(i,x)}\right| = \sum_{\emptyset \neq B \subseteq R(\alpha)} (-1)^{|B|+1} \left|\bigcap_{(i,x) \in B} A_{(i,x)}\right| = \sum_{\emptyset \neq B \subseteq R(\alpha)} (-1)^{|B|+1} \countop(h,B),\]
which concludes the proof.
\end{proof}

\paragraph{Examples.} Examples of the execution of the algorithm for a binary and a ternary string can be found in Figures~\ref{fig:ex} and~\ref{fig:ex2}, respectively.

\begin{figure}[t]
\begin{center}
\input{_fig}
\end{center}
\caption{Computation of $(2,3)$-mappability for the string $T=\texttt{aababba}$ from Example~\ref{ex}. Note that the alphabet is binary in this case, so wildcard subtrees do not need to be introduced. Edges leading to heavy children are drawn in bold. The only substitutions are from a light child to a heavy child. The letters shown above are the original letters before the substitutions. The pairs of compatible modified substrings are indicated with arrows; in the end, $A^3_{=2}[1]=A^3_{=2}[2]=1$ and $A^3_{=2}[3]=A^3_{=2}[4]=A^3_{=2}[5]=2$ as expected.}\label{fig:ex}
\end{figure}

\begin{figure}[t]
\begin{center}
\input{_fig3}
\end{center}
\caption{Computation of $(1,2)$-mappability for the string $T=\texttt{aabaca}$. This example illustrates the use of wildcard symbols.  We have $A^2_{=1}[1]=4$ and $A^2_{=1}[2]=A^2_{=1}[3]=A^2_{=1}[4]=A^2_{=1}[5]=2$.}\label{fig:ex2}
\end{figure}

\paragraph{Correctness.} Let us start with an observation that lists some basic properties of our algorithm. Both parts can be shown by straightforward induction.

\begin{observation}\label{obs:basic}
\begin{enumerate}[label={\rm(\alph*)}]
\item\label{it:obs:basic:a} If a node $v$ stores modified substrings $\alpha,\beta\in \MS(v)$, then it has a descendant $v'$ with $\mathbf{D}(v')=\lcp(\val(\alpha),\val(\beta))$ and $\alpha,\beta\in \MS(v')$.
\item\label{it:obs:basic:b} Every node stores at most one modified substring originating from the same substring $T_{\ell}^m$.
\end{enumerate}
\end{observation}

The following lemma shows that Algorithm~\ref{algo} correctly computes the mappability table $\MP_{=k}^m$.

\begin{lemma}\label{lem:corr}
If $d_H(T_i^m,T_j^m)=k$, then there is exactly one leaf $v$ and exactly one pair of compatible modified substrings $\alpha,\beta\in \MS(v)$ with $i=\idx(\alpha)$ and $j=\idx(\beta)$.
Otherwise, there is no such leaf $v$ and pair $\alpha,\beta$.\end{lemma}
\begin{proof}
Suppose that $\alpha,\beta\in \MS(v)$ are compatible, $i=\idx(\alpha)$, and $j=\idx(\beta)$. Since $W^{-1}(\alpha) \cap W^{-1}(\beta) = \emptyset$, we conclude that $T_i^m$ and $T_j^m$ differ at positions of modifications in $W(\alpha) = W(\beta)$. They differ at positions of modifications in $H(\beta)$ since at the nodes corresponding to these positions, an ancestor of $\alpha$ (that is, the modified substring from which $\alpha$ originates) was in the heavy child and an ancestor of $\beta$ originated from a light child (recall that~\eqref{eq:cond} includes $H(\alpha)\cap H(\beta)=\emptyset$). Symmetrically, $T_i^m$ and $T_j^m$ differ at positions of modifications in $H(\alpha)$. In conclusion, they differ at positions of modifications in $H(\alpha) \cup H(\beta) \cup W(\alpha)$.
The three sets are disjoint, so $|H(\alpha) \cup H(\beta) \cup W(\alpha)|=|H(\alpha)| + |H(\beta)| + |W(\alpha)|=k$ by~\eqref{eq:cond}. This shows that $d_H(T_i^m,T_j^m) \ge k$. With $\val(\alpha)=\val(\beta)$, we conclude that $d_H(T_i^m,T_j^m) = k$.

For a proof in the other direction, assume that $d_H(T_i^m,T_j^m)=k$ and let $1 \le x_1 < x_2 < \cdots < x_k \le m$ be the indices where the two substrings differ. Further let $x_{k+1}=m+1$.
  
  First of all, let us show that there is at least one leaf that contains compatible modified substrings $\alpha$ and $\beta$ with $\idx(\alpha)=i$ and $\idx(\beta)=j$. 
  
  \begin{claim}
    For every $p\in \{1,\ldots,k+1\}$, there exist a node $v_p$ and modified substrings $\alpha_p,\beta_p \in \MS(v_p)$ such that:
    \begin{itemize}
      \item $\idx(\alpha_p)=i$ and $\idx(\beta_p)=j$;
      \item $\lcp(\val(\alpha_p),\val(\beta_p))=x_p-1=\mathbf{D}(v_p)$;
      \item for each position $x_1,\ldots,x_{p-1}$, both $M(\alpha_p)$ and $M(\beta_p)$ contain modifications of wildcard origin,
      or exactly one of these sets contains a modification of heavy origin;
      \item there are no other modifications in $M(\alpha_p)$ or $M(\beta_p)$.
    \end{itemize}
  \end{claim}
  \begin{proof}[of Claim]
    The proof goes by induction on $p$. As $\alpha_1$ and $\beta_1$, we take modified substrings such that $\idx(\alpha_1)=i$, $\idx(\beta_1)=j$, and $M(\alpha_1)=M(\beta_1)=\emptyset$.
    They are stored in the set $\MS(r)$ for the root $r$, so Observation~\ref{obs:basic}\ref{it:obs:basic:a} guarantees the existence of a node $v_1$ with $\mathbf{D}(v_1)=\lcp(\alpha_1,\beta_1)$ and $\alpha_1,\beta_1\in \MS(v_1)$. 
    
    Let $p>1$. By the inductive hypothesis, the set $\MS(v_{p-1})$ contains modified substrings $\alpha_{p-1}$ and $\beta_{p-1}$. The node $v_{p-1}$ has children $w_1$, $w_2$ corresponding to letters $T_i^m[x_{p-1}]$ and $T_j^m[x_{p-1}]$, respectively. If $w_1$ is the heavy child, then $w_2$ is a light child and a modified substring $\beta'$ such that $\idx(\beta')=j$ and $M(\beta')=M(\beta_{p-1}) \cup \{(x_{p-1},T_i^m[x_{p-1}])\}$ is created for the recursive call in $w_1$. Then, we take $\alpha'=\alpha_{p-1}$. The case that $w_2$ is the heavy child is symmetric. Finally, if both $w_1$ and $w_2$ are light children, a child $u$ of $v_{p-1}$ is created along the wildcard symbol \$. There exist modified substrings $\alpha',\beta' \in \MS(u)$ such that: $\idx(\alpha')=i$, $\idx(\beta')=j$, $M(\alpha')=M(\alpha_{p-1})\cup\{(x_{p-1},\$)\}$, and $M(\beta')=M(\beta_{p-1})\cup\{(x_{p-1},\$)\}$.
    
    In either case, we have $\lcp(\val(\alpha'),\val(\beta'))=x_p-1$. The set $(M(\alpha')\cup M(\beta')) \setminus (M(\alpha_{p-1}) \cup M(\beta_{p-1}))$ contains either a modification of heavy origin in one of the modified substrings or modifications of wildcard origin in both. Hence, by the inductive hypothesis, we can set $\alpha_p=\alpha'$ and $\beta_p=\beta'$. The node $v_p$ with $\mathbf{D}(v_p)=\lcp(\val(\alpha_p),\val(\beta_p))$ and $\alpha_p,\beta_p \in \MS(v_p)$
    must exist due to Observation~\ref{obs:basic}\ref{it:obs:basic:a}.
   \end{proof}
  
  Applied for $p=k+1$, the claim yields a leaf $v_{k+1}$ that contains compatible modified substrings $\alpha=\alpha_{k+1}$ and $\beta=\beta_{k+1}$.
  
  Now it suffices to check that there is no other pair of compatible modified substrings $(\alpha',\beta') \ne (\alpha,\beta)$ that would be present in some leaf $u$ and satisfy $\idx(\alpha')=i$ and $\idx(\beta')=j$. Let us first note that $M(\alpha') \cup M(\beta')$ must contain modifications at positions $x_1,\ldots,x_k$ (since $\val(\alpha')=\val(\beta')$) and no modifications at other positions (otherwise, $|H(\alpha')|+|H(\beta')|+|W(\alpha')|$ would exceed $k$). Let $p$ be the greatest index in $\{1,\ldots,k+1\}$ such that $x_p-1 \le \lcp(\val(\alpha),\val(\alpha'))$. By Observation~\ref{obs:basic}\ref{it:obs:basic:b}, $u \ne v_{k+1}$, so $p\le k$.
  
  Thus, the node $v_p$ is an ancestor of the leaf $u$, but the node $v_{p+1}$ is not. Let us consider the children $w_1$, $w_2$ of $v_p$ obtained by following edges with labels $T_i^m[x_{p}]$ and $T_j^m[x_{p}]$, respectively. If $w_1$ is the heavy child, $\beta'$ must contain a modification of heavy origin at position $x_{p}$, so $v_{p+1}$ is an ancestor of $u$; a contradiction. The same contradiction is obtained in the symmetric case that $w_2$ is the heavy child. Finally, if both $w_1$ and $w_2$ are light, then either both $\alpha'$ and $\beta'$ contain a modification of wildcard origin at position $x_{p}$, which again gives a contradiction, or they both contain a modification of heavy origin, which contradicts the first part of condition~\eqref{eq:cond}.
 \end{proof}

\begin{remark}
  The recursive approach presented above is somewhat similar to the scheme used by Thankachan et al.~\cite{DBLP:journals/jcb/ThankachanAA16} for computing the longest common substring with up to $k$ mismatches of two strings. We attempted to adapt the approach of~\cite{DBLP:journals/jcb/ThankachanAA16} to computing $k$-mappability, but failed due to multiple counting of substring pairs, e.g., for $T=\mathtt{aabbab}$, $k=2$, $m=3$. Another virtue of our approach is that we obtain time complexity better by a factor of $k!$ for super-constant $k$.
\end{remark}

\paragraph{Implementation and complexity.}
Our Algorithm~\ref{algo}, excluding the counting phase in the leaves, has exactly the same structure as Algorithm~1 in~\cite{DBLP:conf/cpm/Charalampopoulos18}. Proposition~13 from~\cite{DBLP:conf/cpm/Charalampopoulos18} provides a bound on the total number of the generated modified strings and an efficient implementation based on finger-search trees. We apply that proposition for a family $\mathbf{F}$ composed of substrings $T_i^m$ to obtain the following bounds.

\begin{fact}[{see~\cite[Proposition~13]{DBLP:conf/cpm/Charalampopoulos18}}]\label{fct:time}
  Algorithm~\ref{algo} applied up to the leaves takes\linebreak $\Oh(n\binom{\log n+k+1}{k+1}2^k)$ time
  and generates $\Oh(n\binom{\log n+k}{k}2^k)$ modified substrings.
\end{fact}

Let us further analyze the space complexity of the algorithm.

\begin{lemma}\label{lem:space}
Algorithm~\ref{algo} applied up to the leaves uses $\Oh(nk)$ working space.
\end{lemma}
\begin{proof}
We assume that, upon termination, the procedure \FR discards the set $\MS(v)$
and all the modified strings created during its execution.
This way, the whole memory allocated within a given call to \FR
is freed.
Since \FR returns no output and its only side effects are updates of the array $A_{=}^k$,
no information is lost through such garbage collection.

A call to $\FR(v)$ for node $v$ partitions the list $\MS(v)$ into sublists corresponding to $u_1,\dots,u_a$, creates $2(|\MS(u_2)|+\cdots+|\MS(u_a)|)$ new modified substrings (each requiring constant space to be stored), appends them to sublists corresponding to $u_1$ and $u_{a+1}$, and then recurses on the sublists.
In particular, the elements of the original list $\MS(v)$ are not copied but reused in the recursive call.
The following observation provides further characterization of these elements:

\begin{observation}\label{obs:temp}
If a node $v$ is a child of $w$, then every element of $\MS(v)$ is either an element of $\MS(w)$ or a modified substring originating from an element of $\MS(w)$.
\end{observation}

Let us consider a root-to-leaf path $\rho$ in the recursion. Each recursive call uses $\Oh(1)$ local variables, which take $\Oh(n)$ space overall. 
We also need to bound the total number of modified substrings created by calls to \FR for nodes on the path~$\rho$.

By Observations~\ref{obs:temp} and~\ref{obs:basic}\ref{it:obs:basic:b}, $|\MS(v)|$ is non-increasing on $\rho$.
Moreover, if $v$ is a light child of its parent $w$, then $|\MS(v)|\leq |\MS(w)|/2$.
Let us consider all nodes $w$ on $\rho$ such that the unique child of $w$ that is on $\rho$ is a light child.
The total number of  modified strings created by the calls to $\FR(w)$ for all such nodes $w$ is $\Oh(n)$ since we can upper bound it by a geometric series that sums to $\Oh(n)$.

As for the calls to $\FR(w)$ for the remaining nodes on $\rho$, for every two modified strings they create, they put one of them in the child of $w$ that also belongs to $\rho$.
Hence, it suffices to upper bound the total number of modified substrings originating from $T_i^m$ for each position $i$ that are in $\MS(v)$ for some node $v$ on $\rho$.
For a given position $i$, let $\alpha_1,\dots,\alpha_b$ be all such modified substrings originating from $T_i^m$.
By Observation~\ref{obs:temp}, we have $M(\alpha_1) \subsetneq M(\alpha_2) \subsetneq \dots \subsetneq M(\alpha_b)$ and thus $b \le k$.
In total, we create $\Oh(nk)$ modified substrings in calls to $\FR$ on nodes of $\rho$.
\end{proof}

Next, we show how to improve the time complexity of Algorithm~\ref{algo} by a relatively small change in its execution. Intuitively, we will take advantage of the fact that the modified substrings in a leaf of the recursion do not need to be sorted lexicographically.

Namely, whenever a modified substring $\beta$ with exactly $k$ modifications is created at a node $v$ (i.e., $|M(\alpha)|=k-1$ in the if-statement), we do not include $\beta$ in the recursive call of \MergedTree or \HeavySon. Instead, an entry $(\val(\beta),\beta)$ is inserted into a global hash table. When processing a leaf $v$ containing  modified substrings with a common value $\val(\alpha)$, we need to move all modified substrings with value $\val(\alpha)$ from the global hash table to the set $\MS(v)$.
Finally, if any modified string $\beta$ created while processing a given node $v$
remains in the hash table upon completion of $\FR(v)$, then $\beta$ is removed from the hash table together with all other modified substrings with the value $\val(\beta)$. At this moment, an artificial leaf of the recursion containing all these modified substrings is created and the standard routine is applied to process this leaf.

Recall that the hash table uses Karp--Rabin fingerprints to index strings and collisions could incur incorrect results in the algorithm. To tackle this issue, whenever a modified substring $\beta$ is inserted to the hash table and there is another modified substring with the same hash in the table, we pick any one such modified substring $\alpha$ and check if $\val(\alpha)=\val(\beta)$ in $\Oh(k)$ time using $\lce$ queries on $T$ with a method that resembles kangaroo jumping~\cite{DBLP:journals/tcs/GalilG87,DBLP:journals/tcs/LandauV86} (it requires $\Oh(n)$-time preprocessing). By Lemma~\ref{lem:space}, the hash table contains up to $\Oh(nk)$ entries at any given time,
so the collision probability is $\Oh(nk \cdot n^{-C})=\Oh(n^{-C+2})$. Setting $C>c+2$,
we can make sure that this is dominated by the probability that the hash table fails to process the underlying insertion in $\Oh(1)$ time.

Let us call the resulting algorithm Algorithm~1'.

\begin{lemma}\label{lem:time}
The outputs of Algorithms~\ref{algo} and 1' are the same. Moreover, Algorithm~1' works in time $\Oh(n\binom{\log n+k}{k}2^kk)$ with high probability (up to the leaves) and uses the same amount of space as Algorithm~\ref{algo}.
\end{lemma}
\begin{proof}
Let $v$ be a leaf in the recursion of Algorithm~\ref{algo}. If $\MS(v)$ contains at least one modified substring with up to $k-1$ modifications, $v$ will be identified by the recursive procedure of Algorithm~1'. Then, all modified substrings with exactly $k$ modifications that belong to $v$ are populated from the global hash table. If $\MS(v)$ does not contain any modified substring with less than $k$ modifications, $v$ will be identified upon a deletion from the global hash map at the lowest internal node $u$ of the recursion in which a modified substring belonging to $\MS(v)$ was created. Here, we use the fact that the path-labels $\Label(u)$ of all nodes $u$ of the recursion are different. This shows that indeed the leaves of the recursion of Algorithms~\ref{algo} and 1' are the same.

As for the time complexity, the total number of modified substrings created by Algorithm~1' is the same as in Algorithm~\ref{algo}, i.e., $\Oh(n\binom{\log n+k}{k}2^k)$ by Fact~\ref{fct:time}. However, the time necessary to conduct the whole recursive procedure corresponds to the time complexity of Algorithm~\ref{algo} that is run with $k-1$ instead of $k$, i.e., also $\Oh(n\binom{\log n+k}{k}2^k)$ by Fact~\ref{fct:time}. After $\Oh(n)$-time preprocessing, for each modified substring, we can compute its Karp--Rabin fingerprint and check collisions in $\Oh(k)$ time; this accounts for the additional factor $k$ in the time complexity.

Finally, the space complexity stays the same because modified substrings with exactly $k$ modifications are removed from the hash table at latest when the recursion rolls back.
\end{proof}

Lemmas~\ref{lem:space} and~\ref{lem:time} yield the complexity of Algorithm~1'.
Note that, due to the application of the inclusion-exclusion principle in the leaves, we need to multiply the time complexity of the algorithm by $2^k$ and increase the space complexity by $\Oh(n2^kk)$.

\begin{theorem}\label{the:hp}
There exists a Las-Vegas randomized algorithm that computes the $(k,m)$-mappability of a given length-$n$
string in $\Oh(n2^kk)$ space and, with high probability, in $\Oh(n\binom{\log n+k}{k}4^kk)$ time.
For $k = \Oh(1)$, the space is $\Oh(n)$ and the time becomes $\Oh(n \log^{k} n)$.
\end{theorem}

\section{All-Pairs Hamming Distance Problem}\label{sec:2.5}
We will show how the previous algorithm can be modified to solve the all-pairs Hamming distance problem, at the cost of an additional $\log r$-factor in the complexity. We run the algorithm from the previous section for $T$ being a concatenation of all the strings in $\R$ and only with substrings $\{T_i^m\,:\,i \text{ mod } m=1\}$ in the root. The algorithm needs to be updated only at the leaves of the compact trie. Henceforth, let us consider a trie leaf $v$ with a set $\MS(v)=\{\beta_1,\ldots,\beta_p\}$ of modified substrings. We will further denote this set as $\MS$ ($|\MS|=p$). Our goal is to list, for every $\beta \in \MS$, all $\beta' \in \MS$ that are compatible with $\beta$.

\newcommand{\sset}{\mathit{set}}

Let us construct a static balanced binary search tree (BST) in which the leaves correspond to the modified substrings $\beta_i$. In this way, each node of the BST corresponds to a set of subsequent candidates from the leaves of its subtree. If $\beta_i,\ldots,\beta_j$ are the modified substrings in the leaves of the subtree of a BST node $u$, then we denote $\sset(u) = \{\beta_i,\ldots,\beta_j\}$. A leaf will be responsible for storing information only for itself and an internal node stores merged information of its children.

Our goal is to store information in each node $u$ of the BST in such a way that for any modified substring $\alpha \in \MS$ we will be able to answer if there is any other candidate in $\sset(u)$ that is compatible with $\alpha$. Therefore, in each node $u$, we will compute all the required machinery for using the inclusion-exclusion principle on the modified substrings in $\sset(u)$, that is, a hashmap that stores all non-zero values of $\countop(s,B)$ for modified substrings $\beta \in \sset(u)$. Since every $\beta \in \MS$ is present in $\Oh(\log p)$ sets $\sset(u)$, precomputing all mentioned information can be done in $O(2^kk p \log p)$ time and space.

Our query algorithm for a given modified substring $\beta$ is a recursive procedure starting at the root of the BST.\@ Assume that the algorithm is at some BST node $u$. We use Lemma~\ref{lem:inc_exc} and the hashmap for $\sset(u)$ to count the elements $\beta' \in \sset(u)$ that are compatible with $\beta$. If this number is positive, the algorithm recursively descends to the children of node $u$. In the end, modified substrings $\beta'$ that are compatible with $\beta$ will be listed at the leaves of the BST.\@ The correctness of this algorithm follows from Lemma~\ref{lem:corr}. 

Every application of Lemma~\ref{lem:inc_exc} takes $\Oh(2^kk)$ time. For each modified substring $\beta'$ that is compatible with a modified substring $\beta$, the algorithm will visit $\Oh(\log p)$ BST nodes, which gives $\Oh(2^kk \log p)$ time for finding each compatible modified substring $\beta' \in \MS$. Note that $p \le r$ (see Observation~\ref{obs:basic}\ref{it:obs:basic:b}). Summing up over all trie nodes $v$ and applying Lemmas~\ref{lem:time} and~\ref{lem:space}, we obtain the following result.
(Observe that~\cite[Proposition~13]{DBLP:conf/cpm/Charalampopoulos18} is applied for a family $\mathbf{F}$ of size $r$ rather than $n$.)
  
\begin{theorem}\label{the:hd}
There exists a Las-Vegas randomized algorithm that, given a set of $r$ length-$m$ strings and an integer $k$, 
solves the all-pairs Hamming distance problem in $\Oh(rm+2^kk r\log r)$ space and, with high probability, in $\Oh(rm + r\binom{\log r+k}{k}4^kk\log r+ \mathsf{output}\,\cdot 2^kk \log r)$ time.
For $k = \Oh(1)$, the space is $\Oh(rm+r\log r)$ and the time becomes $\Oh(rm + r \log^{k+1} r+ \mathsf{output}\cdot\log r)$.
\end{theorem}

\section{Computing Mappability in $\Oh(nm^k)$ Time and $\Oh(n)$ Space}\label{sec:2}
In this section, we generalize the $\Oh(nm)$-time algorithm for $k=1$ and integer alphabets from~\cite{DBLP:conf/cocoa/AlzamelCIPRS17}. To this end, we make use of an approach from~\cite{DBLP:conf/spire/AyadBCIP18}.
The high-level idea from~\cite{DBLP:conf/spire/AyadBCIP18} is to define a lexicographic order
on the suffixes of $T$ that ignores the same $k$ fixed positions of every suffix. 
(In fact, the algorithm does the same for many such combinations of $k$ positions.) It then uses the suffix tree of $T$ to sort the modified suffixes according to this new lexicographic order. 
The focus of this algorithm is not on counting substrings that are at Hamming distance at most $k$, and so we adapt it with some extra care to avoid multiple counting.

We first generate all $\binom{m}{\le k}$ subsets of $\{1, \ldots , m \}$ of size at most $k$.
For each such subset $F$, we consider the length-$m$ substrings of $T$ with their $f$-th letter substituted with $\# \not\in \Sigma$ for all $f \in F$.
We sort each of these sets of strings in $\Oh(nk\binom{m}{\le k})$ total time using the approach of~\cite{DBLP:conf/spire/AyadBCIP18}, also obtaining the maximal blocks of equal strings in the sorted list.

We now briefly describe the algorithm for sorting one such set of strings in time $\Oh(nk)$ for the sake of completeness. Let us assume for simplicity that $F=\{f\}$ as the algorithm can be generalized trivially for larger sets.
We first retrieve the sorted list of $T_i^{f-1}$ for all $i$ from the suffix tree. We then give ranks to these strings after we check equality of adjacent strings in the sorted list using $\lce$ queries.
We similarly rank strings $T_{j}^{m-f}$ for all $j$.
Finally, we sort the ranks of the pairs $(T_i^{f-1}, T_{i+f+1}^{m-f})$ using bucket sort.

Prior to running the above algorithm, we initialize arrays $D_K$ for $K \in \{1, \ldots, k\}$.
For each maximal block, of size $b$, of equal strings obtained for some set $F$, we increment the $b$ relevant entries of $D_{|F|}$ by $b-1$.

Note that if $d_H(T_i^m,T_j^m)=\kappa$, then this will contribute $\binom{m-\kappa}{K-\kappa}$ to each of $D_{K}[i]$ and $D_{K}[j]$ for $K\ge\kappa$, since there are these many size-$K$ supersets of the set of mismatching positions in the power set of $\{1, \ldots , m \}$.
We thus compute $\MP_{={K}}^m[i]=D_K[i]-\sum_{\kappa=0}^{K-1} \binom{m-\kappa}{K-\kappa} \MP_{=\kappa}[i]$ in increasing order with respect to $K$ and we are done. (We precompute all relevant binomial coefficients in $\Oh(k^2)$ time.)

\begin{theorem}\label{the:nmk}
Given a string of length $n$, the $(k,m)$-mappability problem can be solved in $\Oh(nk\binom{m}{\le k})$ time and $\Oh(n)$ space. For $k= \Oh(1)$, the time becomes $\Oh(n m^k)$.
\end{theorem}

Combining~\cref{the:nmk,the:hp} gives the following result.

\begin{corollary}
For every $k=\Oh(1)$, there exists a randomized algorithm that computes the $(k,m)$-mappability of a given length-$n$ string in $\Oh(n)$ space and in $\Oh(n \cdot \min \{m^k,\log^k n\})$ time with high probability.
\end{corollary}

\section{Computing $(k,m)$-Mappability for All $k$ or for All $m$}\label{sec:3}
\begin{theorem}
The $(k,m)$-mappability for a given $m$ and all $k\in \{0,
\ldots, m\}$ can be computed in $\Oh(n^2)$ time using $\Oh(n)$ space.
\end{theorem}
\begin{proof}
  We first present an algorithm which solves the problem
  in $\Oh(n^2)$ time using $\Oh(n^2)$
  space and then show how to reduce the space usage to $\Oh(n)$.
  
  We initialize an $n \times n$ matrix $M$ in which $M[i,j]$ will store the Hamming distance
  between substrings
  $T_i^m$ and $T_j^m$. Let us consider two letters $T[i]\neq T[j]$ of the  input string, where $i < j$. Such a pair contributes to
  a mismatch between the following 
  pairs of strings:
  \[(T_{i - m + 1}^m, T_{j - m + 1}^m),
  (T_{i - m + 2}^m, T_{j - m + 2}^m), \ldots ,
  (T_i^m, T_j^m).\] 
  This list of strings is represented by a
  diagonal interval in $M$, the entries of which we need to increment by $1$. We process all $\Oh(n^2)$ pairs of letters and
  update the information
  on the respective intervals. Then $\MP_{= k}^m[i] = |\{j\,:\,M[i,j]= k\}|$.
 
  To achieve $\Oh(1)$ time for each single addition on
  a diagonal interval,
  we use a well-known trick from an analogous problem in one dimension.
  Suppose that we would like to add $1$ on the diagonal interval from
  $M[x_1, y_1]$ to $M[x_2, y_2]$. Instead, we can simply add $1$ to $M[x_1,
  y_1]$ and $-1$ to 
  $M[x_2 + 1, y_2 + 1]$. Every cell will then
  represent the difference of its actual value to the actual value of its predecessor on the diagonal. After
  all such operations are performed, we can
  retrieve the actual values by computing prefix sums on each diagonal in a top-down manner.
  
  To reduce the space usage to $\Oh(n)$, it suffices to observe that the value of $M[i, j]$ depends only on the value of $M[i-1,j-1]$ and at most two
  letter comparisons which can add $+1$ and/or $-1$ to the cell. Recall that $M[i, j]=d_H(T_i^m,T_j^m)$. We need to subtract 1 from the previous result if the first characters of the previous substrings were equal and add 1 if the last characters of the new substrings were different. Therefore, we can process the matrix row by row, from top to bottom, and compute the values $\MP_{=0}^m[i],\ldots,\MP_{={m}}^m[i]$ while processing the $i$th row. 
  \end{proof}

\begin{theorem}\label{thm:allm}
The $(k,m)$-mappability for a given $k$ and all $m\in \{k,\ldots, n\}$ can be computed in $\Oh(n^2)$ time and space.
\end{theorem} 
\begin{proof}
We first prove the following claim.
\begin{claim}
The longest common prefixes with $k$ mismatches for all pairs of suffixes of $T$ can be computed in $\Oh(n^2)$ time.
\end{claim}
\begin{proof}[of Claim]
We process the pairs in batches $B_\delta$ for $\delta\in \{1,2,\ldots,n\}$ so that the pair $(T_i,T_j)$, which we denote by $(i,j)$, is in $B_{|j-i|}$. It now suffices to show how to process a single batch $B_\delta$ in $\Oh(n)$ time. We will do so by comparing pairs of letters of $T$ at distance $\delta$ from left to right.
We first compute $\lce_k(1,1+\delta)$ naively. Then, given that $\lce_k(i,j)=\ell$, where $j-i=\delta$, we will retrieve $\lce_k(i+1,j+1)$ using the following simple observation: either $j+\ell-1=n$, or $T_i^{\ell}$ and $T_j^{\ell}$ have exactly $k$ mismatches and $T[i+\ell] \neq T[j+\ell]$. In the former case, we trivially have that $\lce_k(i+1,j+1)=\ell-1$. In the latter case, we first check whether $T[i]=T[j]$, in which case $d_H(T_{i+1}^{\ell-1},T_{j+1}^{\ell-1})=k$ and hence $\lce_k(i+1,j+1)=\ell-1$. If $T[i]\neq T[j]$, then $d_H(T_{i+1}^{\ell-1},T_{j+1}^{\ell-1})=k-1$ and we perform letter comparisons to extend the match. The pairs of letters compared in this step have not been compared before; the complexity follows. 
\end{proof}

We store the information on $\lce_k$'s as follows. We initialize an $n \times n$ matrix $Q$. Then, for a pair $(i,j)$ such that $\lce_k(i,j)=\ell$, we increment by $1$ the entries $Q[\ell,i]$ and $Q[\ell,j]$.
Note that if $\lce_k(i,j)=\ell$, then $i$ (resp.~$j$) will contribute $1$ to the $(k,m)$-mappability values $\MP^m_{\le k}[j]$ (resp.~$\MP^m_{\le k}[i]$) for all $m \in \{k, \ldots, \ell\}$.
Thus, starting from the last row of $Q$, we iteratively add row $\ell$ to row $\ell-1$. In the end, by the above observation, row $m$ stores the $(k,m)$-mappability array $\MP^m_{\le k}$. 
\end{proof}

\section{Conditional Hardness for $k,m = \Theta(\log n)$}\label{sec:4}
We will show that $(k,m)$-mappability cannot be computed in strongly
subquadratic time in case that the parameters are $\Theta(\log n)$, unless
the Strong Exponential Time Hypothesis (SETH) of Impagliazzo, Paturi and
Zane~\cite{DBLP:journals/jcss/ImpagliazzoPZ01,DBLP:journals/jcss/ImpagliazzoP01} fails. 
Our proof is based on the conditional hardness of the following decision version of the Longest Common Substring with $k$ Mismatches problem.

\defproblem{Common Substring of Length $d$ with $k$
Mismatches}{
Strings $T_1, T_2$ of length $n$ over binary alphabet and integers $k$, $d$.}{Is there a factor of $T_1$ of length $d$ that occurs in
$T_2$ with $k$ mismatches?}

\begin{lemma}[\cite{DBLP:journals/corr/abs-1712-08573}]\label{lem:LCF_hard}
  Suppose there is $\varepsilon > 0$ such that Common Substring of Length $d$ with $k$
  Mismatches can be solved in $\Oh(n^{2-\varepsilon})$ time on strings
  over binary alphabet for $k = \Theta(\log n)$ and $d=21k$. Then SETH is
  false.
\end{lemma}

\begin{theorem}
  If the $(k,m)$-mappability can be computed in $\Oh(n^{2-\varepsilon})$
  time for binary strings, $k,m=\Theta(\log n)$, and
  some $\varepsilon>0$, then SETH is false.
\end{theorem}
\begin{proof}
  We make a Turing reduction from Common Substring of Length $d$ with $k$ Mismatches.
  Let $T_1$ and $T_2$ be the input to the problem. We compute the
  $(k,d)$-mappabilities of strings $T_1\cdot T_2$ and $T_1\cdot T_2[1 \dd d-1]$ and store them in arrays $A$
  and $B$, respectively. For each $i\in \{1,\ldots,n-d+1\}$, we
  subtract $B[i]$ from $A[i]$. Then, $A[i]$ holds the number of
  factors of $T_2$ of length $d$ that are at Hamming distance $k$ from
  $T_1[i \dd i+d-1]$. Hence, Common Substring of Length $d$ with $k$ Mismatches has a positive answer if and only if $A[i]>0$ for any $i\in \{1,\ldots,n-d+1\}$.
  
  By Lemma~\ref{lem:LCF_hard}, an $\Oh(n^{2-\varepsilon})$-time algorithm
  for Common Substring of Length $d$ with $k$ Mismatches with $k=\Theta(\log n)$ and
  $d=21k$ would refute SETH.\@ By the shown reduction, an
  $\Oh(n^{2-\varepsilon})$-time algorithm for $(k,m)$-mappability with
  $k,m = \Theta(\log n)$ would also refute SETH.
\end{proof}

\section{Final Remarks}

Our main contribution is an $\Oh(n \cdot \min\{m^k,\log^k n\})$-time $\Oh(n)$-space algorithm for solving the $(k,m)$-mappability problem. Let us recall that genome mappability, as introduced in~\cite{biopaper}, counts the number of substrings that are at Hamming distance at most $k$ from every length-$m$ substring of the text. One may also be interested to consider mappability under the edit distance model. This question relates also to recent contributions on computing approximate longest common prefixes and substrings under edit distance~\cite{DBLP:conf/recomb/ThankachanACA18,DBLP:conf/spire/AyadBCIP18}. In the case of the edit distance, in particular, a decision needs to be made whether sufficiently similar substrings only of length exactly $m$ or of all lengths between $m-k$ and $m+k$ should be counted. We leave the mappability problem under edit distance for future investigation.

\bibliographystyle{plainurl}
\bibliography{references}

\end{document}

%% file: _fig.tex
\definecolor{change}{rgb}{1,.3,.3}

\newcommand{\dodo}[3]{
\begin{scope}[yshift=-.07cm]
\draw (0.08,0) node[anchor=mid]{\tt #1};
  \draw (-.25, -0) node[anchor=mid]{#2};
  \draw (0.08,.23) node[anchor=mid,text=change]{\tt{#3}~};
  \end{scope}
}

\setlength{\unit}{.725cm}
\setlength{\nodemargin}{.1cm}
\setlength{\nodesep}{.42cm}
\setlength{\level}{1.5cm}

\begin{tikzpicture}
{\footnotesize

\begin{scope}[xshift=6.5\unit+3\nodemargin+2.5\nodesep]
  \draw (0,0) node[draw,rectangle,minimum width=5\unit+\nodemargin, minimum height=.65cm](eps) {};
  \begin{scope}[xshift=-2\unit]\dodo{aab}{1}{~~~}\end{scope}
  \begin{scope}[xshift=-\unit]\dodo{aba}{2}{~~~}\end{scope}
  \begin{scope}[xshift=0]\dodo{abb}{4}{~~~}\end{scope}
  \begin{scope}[xshift=\unit]\dodo{bab}{3}{~~~}\end{scope}
  \begin{scope}[xshift=2\unit]\dodo{bba}{5}{~~~}\end{scope}
  \end{scope}
  
  \begin{scope}[xshift=4.5\unit+1.5\nodemargin+1\nodesep,yshift=-\level]
     \draw (0,0) node[draw,rectangle,minimum width=5\unit+\nodemargin, minimum height=.65cm](a) {};
      
  \begin{scope}[xshift=-2\unit]\dodo{aab}{1}{~~~}\end{scope}
  \begin{scope}[xshift=-\unit]\dodo{aab}{3}{b~~}\end{scope}
  \begin{scope}[xshift=0]\dodo{aba}{2}{~~~}\end{scope}
  \begin{scope}[xshift=\unit]\dodo{aba}{5}{b~~}\end{scope}
  \begin{scope}[xshift=2\unit]\dodo{abb}{4}{~~~}\end{scope}
     
   \end{scope}
   
   -

\begin{scope}[xshift=\unit+.5\nodemargin,yshift=-3\level]
  \draw (0,0) node[draw,rectangle,minimum width=2\unit+\nodemargin, minimum height=.65cm](aab) {};
  \begin{scope}[xshift=-.5\unit]\dodo{aab}{1}{~~~}\end{scope}
  \begin{scope}[xshift=.5\unit]\dodo{aab}{3}{b~~}\end{scope}   
\end{scope}

\begin{scope}[xshift=5.5\unit+2\nodemargin+1.5\nodesep,yshift=-2\level]
  \draw (0,0) node[draw,rectangle,minimum width=5\unit+\nodemargin, minimum height=.65cm](ab) {};
  \begin{scope}[xshift=-2\unit]\dodo{aba}{2}{~~~}\end{scope}
  \begin{scope}[xshift=-\unit]\dodo{aba}{5}{b~~}\end{scope}
  \begin{scope}[xshift=0]\dodo{abb}{1}{~a~}\end{scope}
  \begin{scope}[xshift=\unit]\dodo{abb}{3}{ba~}\end{scope}
  \begin{scope}[xshift=2\unit]\dodo{abb}{4}{~~~}\end{scope}
\end{scope}

\begin{scope}[xshift=3\unit+1.5\nodemargin+\nodesep,yshift=-3\level]
   \draw (0,0) node[draw,rectangle,minimum width=2\unit+\nodemargin, minimum height=.65cm](aba) {};
   \begin{scope}[xshift=-.5\unit]\dodo{aba}{2}{~~~}\end{scope}
   \begin{scope}[xshift=.5\unit]\dodo{aba}{5}{b~~}\end{scope}   
\end{scope}

\begin{scope}[xshift=6.5\unit+2.5\nodemargin+2\nodesep,yshift=-3\level]
  \draw (0,0) node[draw,rectangle,minimum width=5\unit+\nodemargin, minimum height=.65cm](abb) {};
  \begin{scope}[xshift=-2\unit]\dodo{abb}{1}{~a~}\end{scope}
  \begin{scope}[xshift=-\unit]\dodo{abb}{2}{~~a}\end{scope}
  \begin{scope}[xshift=0]\dodo{abb}{3}{ba~}\end{scope}
  \begin{scope}[xshift=\unit]\dodo{abb}{4}{~~~}\end{scope}
  \begin{scope}[xshift=2\unit]\dodo{abb}{5}{b~a}\end{scope}
  \draw (0,-.3) node[below] {$1 \leftrightarrow 2,\quad 3 \leftrightarrow 4,\quad 4 \leftrightarrow 5$};
\end{scope}

\begin{scope}[xshift=11\unit+4.5\nodemargin+4\nodesep,yshift=-\level]
  \draw (0,0) node[draw,rectangle,minimum width=2\unit+\nodemargin, minimum height=.65cm](b) {};
  \begin{scope}[xshift=-.5\unit]\dodo{bab}{3}{~~~}\end{scope}
  \begin{scope}[xshift=.5\unit]\dodo{bba}{5}{~~~}\end{scope}

 \end{scope}

\begin{scope}[xshift=10.5\unit+4\nodemargin+3.5\nodesep,yshift=-2\level]
  \draw (0,0) node[draw,rectangle,minimum width=2\unit+\nodemargin, minimum height=.65cm](ba) {};
  \begin{scope}[xshift=-.5\unit]\dodo{baa}{5}{~b~}\end{scope}
  \begin{scope}[xshift=.5\unit]\dodo{bab}{3}{~~~}\end{scope}
\end{scope}

\begin{scope}[xshift=10\unit+3.5\nodemargin+3\nodesep,yshift=-3\level]
  \draw (0,0) node[draw,rectangle,minimum width=2\unit+\nodemargin, minimum height=.65cm](baa) {};
  \begin{scope}[xshift=-.5\unit]\dodo{baa}{3}{~~b}\end{scope}
  \begin{scope}[xshift=.5\unit]\dodo{baa}{5}{~b~}\end{scope}
 
  \draw (0,-.3) node[below] {$3 \leftrightarrow 5$};
\end{scope}

\begin{scope}[xshift=11.5\unit+4.5\nodemargin+4\nodesep,yshift=-3\level]
  \draw (0,0) node[draw,rectangle,minimum width=\unit+\nodemargin, minimum height=.65cm](bab) {};
  \begin{scope}\dodo{bab}{3}{~~~}\end{scope}
\end{scope}

\begin{scope}[xshift=12.5\unit+5.5\nodemargin+5\nodesep,yshift=-3\level]
  \draw (0,0) node[draw,rectangle,minimum width=\unit+\nodemargin, minimum height=.65cm](bba) {};
  \begin{scope}\dodo{bba}{5}{~~~}\end{scope}
\end{scope}

}

 \draw[-latex, ultra thick] (eps) -- node[near end, above left=-.15cm] {\tt a} (a);
   \draw[-latex] (a) -- node[near end, above left=-.15cm] {\tt ab} (aab);
 
  \draw[-latex, ultra thick] (a) -- node[near end, above right=-.1cm] {\tt b} (ab);
  \draw[-latex] (ab) -- node[near end, above left=-.15cm] {\tt a} (aba);
  \draw[-latex, ultra thick] (ab) -- node[near end, above right=-.1cm] {\tt b} (abb);
  
  \draw[-latex] (eps) -- node[near end, above left=-.0cm] {\tt b} (b);
  \draw[-latex, ultra thick] (b) -- node[near end, above left=-.1cm] {\tt a} (ba);
  \draw[-latex, ultra thick] (ba) -- node[near end, above left=-.1cm] {\tt a} (baa);
  \draw[-latex] (ba) -- node[near end, above right=-.15cm] {\tt b} (bab);
  \draw[-latex] (b) -- node[near end, above right=-.15cm] {\tt ba} (bba);
  
\end{tikzpicture}

%% file: _fig3.tex
\definecolor{change}{rgb}{1,.3,.3}

\newcommand{\dodo}[3]{
\begin{scope}[yshift=-.07cm]
 \draw (0.08,0) node[anchor=mid]{\tt #1};
  \draw (-.16, -0) node[anchor=mid]{#2};
  \draw (0.08,.23) node[anchor=mid,text=change]{\tt{#3}~};
\end{scope}
}
\setlength{\unit}{.55cm}
\setlength{\nodemargin}{.2cm}
\setlength{\nodesep}{.5cm}
\setlength{\level}{1.5cm}
\begin{tikzpicture}
{\footnotesize
 

\begin{scope} [xshift=6.5\unit+3.5\nodemargin+3\nodesep]
  \draw (0,0) node[draw,rectangle,minimum width=5\unit+\nodemargin, minimum height=.65cm](eps) {};
  \begin{scope}[xshift=-2\unit]\dodo{aa}{1}{~~}\end{scope}
  \begin{scope}[xshift=-\unit]\dodo{ab}{2}{~~}\end{scope}
  \begin{scope}[xshift=0]\dodo{ac}{4}{~~}\end{scope}
  \begin{scope}[xshift=\unit]\dodo{ba}{3}{~~}\end{scope}
  \begin{scope}[xshift=2\unit]\dodo{ca}{5}{~~}\end{scope}
  \end{scope}
  
  \begin{scope}[xshift=4.5\unit+2\nodemargin+1.5\nodesep,yshift=-\level]
     \draw (0,0) node[draw,rectangle,minimum width=5\unit+\nodemargin, minimum height=.65cm](a) {};
      
  \begin{scope}[xshift=-2\unit]\dodo{aa}{1}{~~}\end{scope}
  \begin{scope}[xshift=-\unit]\dodo{aa}{3}{b~}\end{scope}
  \begin{scope}[xshift=0]\dodo{aa}{5}{c~}\end{scope}
  \begin{scope}[xshift=\unit]\dodo{ab}{2}{~~}\end{scope}
  \begin{scope}[xshift=2\unit]\dodo{ac}{4}{~~}\end{scope}
   \end{scope}

\begin{scope}[xshift=2.5\unit+.5\nodemargin,yshift=-2\level]
  \draw (0,0) node[draw,rectangle,minimum width=5\unit+\nodemargin, minimum height=.65cm](aa) {};
  \begin{scope}[xshift=-2\unit]\dodo{aa}{1}{~~}\end{scope}
    \begin{scope}[xshift=-\unit]\dodo{aa}{2}{~b}\end{scope}
  \begin{scope}[xshift=0]\dodo{aa}{3}{b~}\end{scope}
  \begin{scope}[xshift=\unit]\dodo{aa}{4}{~c}\end{scope} 
  \begin{scope}[xshift=2\unit]\dodo{aa}{5}{c~}\end{scope}
  \draw (0,-0.3) node[below] {$1 \leftrightarrow 2, 1\leftrightarrow 3, 1 \leftrightarrow 4, 1 \leftrightarrow 5$};
\end{scope}

\begin{scope}[xshift=5.5\unit+1.5\nodemargin+\nodesep,yshift=-2\level]
  \draw (0,0) node[draw,rectangle,minimum width=\unit+\nodemargin, minimum height=.65cm](ab) {};
  \begin{scope}\dodo{ab}{2}{~~}\end{scope}
\end{scope}

\begin{scope}[xshift=6.5\unit+2.5\nodemargin+2\nodesep,yshift=-2\level]
  \draw (0,0) node[draw,rectangle,minimum width=\unit+\nodemargin, minimum height=.65cm](ac) {};
  \begin{scope}\dodo{ac}{4}{~~}\end{scope}
\end{scope}

\begin{scope}[xshift=8\unit+3.5\nodemargin+3\nodesep,yshift=-2\level]
  \draw (0,0) node[draw,rectangle,minimum width=2\unit+\nodemargin, minimum height=.65cm](as) {};
  \begin{scope}[xshift=-.5\unit]\dodo{a\$}{2}{~b}\end{scope}
  \begin{scope}[xshift=.5\unit]\dodo{a\$}{4}{~c}\end{scope} 
    \draw (0,-0.3) node[below] {$2 \leftrightarrow 4$};
  
\end{scope}

\begin{scope}[xshift=9.5\unit+4.5\nodemargin+4\nodesep,yshift=-2\level]
  \draw (0,0) node[draw,rectangle,minimum width=\unit+\nodemargin, minimum height=.65cm](ba) {};
  \begin{scope}\dodo{ba}{3}{~~}\end{scope}

 \end{scope}

 \begin{scope}[xshift=10.5\unit+5.5\nodemargin+5\nodesep,yshift=-2\level]
  \draw (0,0) node[draw,rectangle,minimum width=\unit+\nodemargin, minimum height=.65cm](ca) {};
  \begin{scope}\dodo{ca}{5}{~~}\end{scope}

 \end{scope}
 
 \begin{scope}[xshift=12\unit+6.5\nodemargin+6\nodesep,yshift=-2\level]
  \draw (0,0) node[draw,rectangle,minimum width=2\unit+\nodemargin, minimum height=.65cm](sa) {};
  \begin{scope}[xshift=-.5\unit]\dodo{\$a}{3}{b~}\end{scope}
  \begin{scope}[xshift=.5\unit]\dodo{\$a}{5}{c~}\end{scope} 
  \draw (0,-0.3) node[below] {$3 \leftrightarrow 5$};
 \end{scope}
}

 \draw[-latex, ultra thick] (eps) -- node[near end, above left=-.1cm] {\tt a} (a);
  \draw[-latex, ultra thick] (a) -- node[near end, above left=-.1cm] {\tt a} (aa);
  \draw[-latex] (a) -- node[near end, above right=-.15cm] {\tt b} (ab);
  \draw[-latex] (a) -- node[near end, above right=-.15cm] {\tt c} (ac);
  \draw[-latex] (a) -- node[near end, above right=-.15cm] {\tt \$} (as);

  \draw[-latex] (eps) -- node[near end, above right=-.15cm] {\tt ba} (ba);
    \draw[-latex] (eps) -- node[near end, above right=-.15cm] {\tt ca} (ca);
    \draw[-latex] (eps) -- node[near end, above right=-.15cm] {\tt \$a} (sa);

\end{tikzpicture}